\documentclass[preprint,12pt]{elsarticle}

\usepackage{amssymb}
\usepackage{amssymb}

\usepackage{graphics}
\usepackage{graphicx}
\usepackage{latexsym}
\usepackage{amsmath}

\usepackage{indentfirst}
\usepackage{CJK}
\usepackage{xypic}
\newtheorem{definition}{Definition}[section]

\newtheorem{theorem}[definition]{Theorem}

\newproof{proof}{\textbf{Proof}}

\begin{document}

\begin{frontmatter}



\title{\textbf{ An ideal multi-secret sharing scheme based
              on minimal privileged coalitions }}

\author{Yun Song , Zhihui Li $^{*}$}
\cortext[cor1]{Corresponding author:Zhihui Li.(lizhihui25@yahoo.com.cn)
}

\address{College of Mathematics and Information Science, Shaanxi Normal
University, Xi'an,
710062, P. R. China}

\begin{abstract}
 How to construct an ideal multi-secret sharing scheme for general access structures is
	difficult. In this paper, we solve an open problem proposed by Spiez et al.recently [Finite Fields and
	Their Application, 2011(17) 329-342], namely to design an algorithm of privileged coalitions of any
	length if such coalitions exist. Furthermore, in terms of privileged coalitions, we show that most of the
	existing multi-secret sharing schemes based on Shamir threshold secret sharing are not perfect by analyzing Yang et al.'s scheme and Pang et al.'s scheme. Finally, based on the algorithm mentioned above, we devise an ideal multi-secret
	sharing scheme for families of access structures, which possesses more vivid authorized sets than
	that of the threshold scheme.
\end{abstract}


\end{frontmatter}

\section{Introduction}
\label{intro}\emph{Single-secret sharing schemes(SSSS). } A secret sharing scheme is a method to distribute  a secret among a set $P$ of participants, which includes a pair of efficient algorithms:a distribution algorithm and a reconstruction
algorithm, implemented by a dealer and some participants.
The distribution algorithm allows a dealer to split a secret
s into different pieces, called shares, and distribute them to
participants. The reconstruction algorithm is executed by the
authorized subsets of parties who are able to reconstruct the
secret by using their respective shares. The collection of these
authorized sets of participants is called the access structure,\ $\Gamma\subset 2^{P}$. A group of participants is called a minimal authorized subsets if
they can recover the secret with their shares, and any of its
proper subgroups cannot do so. Then, the access
structure  is determined by the family of minimal authorized subsets, $(\Gamma)_{min}$. The notion of secret sharing was introduced by Shamir [1] and Blakley [2] , who considered the only schemes with a $(t,n)$-threshold access structure
formed by the set of authorized subsets of participants is the set of all subsets of size at least $t$, for some integer $t$. In 1987, Ito et al.[3] proved that there exists a secret
sharing scheme for any access structure which is more general than the threshold ones. Afterwards, secret sharing schemes have been widely employed in the construction of more elaborate
cryptographic primitives and several types of cryptographic protocols(see [4-8]).

A secret sharing scheme is called perfect if any non-authorized subset of
participants have no information about the secret, and ideal if
the shares are of the same size as that of the secret. An ideal secret sharing scheme is well-known to be the best efficiency that one can achieve with lowest storage complexity and the communication complexity[9].

\emph{Multi-secret sharing schemes(MSSS).}  Multi-secret sharing can be seen as a natural generalization of single secret sharing schemes. In 1994, Blundo et al.
[10] studied the more general case in which the set of participants share more than
one secret and different secret is related to different access structure.  Let $\textbf{$\Gamma$}=(\Gamma_{0},\ldots,\Gamma_{m-1})$ be the
m-tuple of access structures on $P$ and let $S_{0}\times S_{1}\times \cdots \times S_{m-1}$be the set from which the secrets are chosen, where for any  $0\leq j\leq m-1$, each secret $s_{j}$ to be shared
is chosen in $S_{j}$.
 In the definition of a perfect multi-secret sharing scheme,  an m-tuple of secrets $(s_{0},\ldots,s_{m-1})\in S_{0}\times\cdots\times S_{m-1}$ is  shared  in an  m-tuple $(\Gamma_{0},\ldots,\Gamma_{m-1})$ of  access
structures on $P$  in such a way that, for each $0\leq j\leq m-1$, the access structure $\Gamma_{j}$
is the set of
all subsets of $P$ that can recover secret $s_{j}.$  A perfect multi-secret sharing scheme is
defined in [10] such that the following requirements are satisfied.

\begin{definition} \label{1.1} Let\ $\textbf{$\Gamma$}=\{\Gamma_{0},\ldots,\Gamma_{m-1}\}$ be an m-tuple of access structures on the set of participants $P=\{P_{1},\ldots,P_{n}\}$. A multi-secret sharing
scheme for \ $\textbf{$\Gamma$}=\{\Gamma_{0},\ldots,\Gamma_{m-1}\}$
is a sharing of the secrets $(s_{0},\ldots,s_{m-1})\in S_{0}\times\cdots\times S_{m-1}$ in such a way that, for $0\leq j\leq m-1$,

\hspace{-0.5cm}(1) Correctness requirement: Any subset $A\subseteq P$ of participants enabled to recover

\hspace{0.2cm}$s_{j}$ can compute $s_{j}$. Formally,
for all $ A\in \Gamma_{j}$, it holds $H(S_{j}\ |\ A)=0,$

\hspace{-0.6cm}(2) Security requirement: Any subset $A\subseteq P$ of  participants not enabled to recover

\hspace{0.2cm}$s_{j}$, even knowing some of the other secrets, has no more information on $s_{j}$

\hspace{0.2cm}than that already conveyed by  the  known secrets. Formally, for all $ A\not\in \Gamma_{j}$ and

 \hspace{0.2cm}$T\subseteq\{S_{0},\cdots, S_{m-1}\}\backslash\{S_{j}\}$, it holds $H(S_{j}\ |\ AT)=H(S_{j}\ |\ T).$
\end{definition}

So far, Multi-secret sharing are widely applied not only in the field of information security but also the theories and models of
secret sharing schemes[11-19].

\emph{Our results.}  Although fruitful results
for the multi-secret sharing have been obtained, it is still difficult to devise ideal multi-secret sharing schemes for general access structures. In this paper, by using the theory of privileged coalitions[20,21], we point out that several multi-secret sharing schemes
(see[13-18])based on  Shamir's threshold scheme are not perfect, and thus are not ideal. In[20] Spiez et al. obtained an algorithm to construct the privileged coalitions of maximal length and put forward how to design an algorithm of  privileged coalitions of any length if such coalitions exist. Motivated by these concers, we solve this open problem and devise an ideal multi-secret sharing scheme(IMSSS) for families of access structures based on the algorithm mentioned above. Finally, we compare our scheme with two additional schemes[13,14] according to their performance analysis.

The rest of the paper is organized as follows: In Section 2, the basic definitions of secret Shamir's secret sharing scheme and  privileged coalitions are reviewed,  an algorithm of  privileged coalitions of any length is designed as well. In Section 3 and 4, based on theories of privileged coalitions and corresponding algorithm, we shall present our ideal multi-secret sharing scheme and make some discussions. Finally, some remarks are given in the conclusion Section.


\section{Preliminaries}
\subsection{Shamir's secret sharing scheme}
\label{prel} In a \textbf{$(t,n)$ threshold secret sharing scheme}(a scheme with $n$ participants and threshold $t$), where $2\leq t\leq n$, a dealer does not disclose a secret data to the participants but only distributes $n$ shares amongest them in such a way that at least $t$ or more participants can collectively efficiently reconstruct the secret but no coalition of less than $t$ participants can obtain noting about the secret.

In the paper, we consider \textbf{Shamir's secret sharing schemes}[20] with the secret placed as a coefficient $a_{i}$ of the scheme polynomial
\ $f(x)=a_{0}+a_{1}x+\cdots+a_{t-1}x^{t-1}$, where $\textbf{a}=(a_{0},\ldots,a_{t-1})\in F^{t}_{q}$ ($q$ is a prime power).  For a fix $f(x)$ and an $j$, such scheme is uniquely defined by a sequence $\textbf{L}=(l_{1},l_{2}\ldots,l_{n})\in F_{q^{*}}^{n}$ of pairwise different public identities, allocated to participants, called in [20] a track. The shares \ $y_{i}=f(l_{i})$ $(1\leq i\leq n)$ assigned to participants are secret.
\\
\textbf{Remark 1.}\ Shamir's $(t,n)$ threshold secret sharing and Shamir's secret sharing are different. A Shamir's $(t,n)$ threshold secret sharing must be a Shamir's secret sharing, but a Shamir's secret sharing may not be a Shamir's $(t,n)$ threshold secret sharing.

In fact, the public identities\ $\textbf{L}$\ define the\ $n\times t$\ matrix $\textbf{A(L)}=(l^{\nu}_{\mu})_{1\leq \mu\leq n
\atop {0\leq \nu\leq t-1}}$ over \ $F_{q}$\ which gives the shares by
\ $\textbf{A(L)a}^{T}=\textbf{y}^{T}$. Since\ $\textbf{L}$ is a track any coalition of $t$ participants determines a $t\times t$ non-singular
Vandermonde submatrix of the matrix\ $\textbf{A(L)}$\ consisting of the corresponding rows of the matrix\ $\textbf{A(L)}$, i.e.,
\\
(i)\  all $t\times t$ submatrices of\ $\textbf{A(L)}$ are non-singular.

A Shamir's secret sharing is a Shamir's $(t,n)$ threshold secret sharing if and only if
\\
(ii) all $(t-1)\times(t-1)$ submatrices of the matrix obtained from the matrix\ $\textbf{A(L)}$

 by removing its $j$-th
column are non-singular.

The track\ $\textbf{L}$ corresponding to such matrix $\textbf{A(L)}$ satisfying two above conditions is called $(t,j)-$admissible[21].
\begin{definition} \label{2.1} Let\ $0\leq j \leq t-1$. If the track\  $\textbf{L}\in F^{n}_{q}$, where $n\geq k$, defines a Shamir's $(t,n)$ threshold secret sharing with the secret placed as a coefficient \ $a_{j}$\ of the scheme polynomial \ $f(x)=a_{0}+a_{1}x+ \cdots +a_{t-1}x^{t-1}$, then \ $\textbf{L}$ is called a $(t,j)-$admissible track.
\end{definition}

Note that if the track\ $\textbf{L}\in F^{n}_{q}$\ is not\ $(t,j)-$admissible, then it contains a subtrack consisting of less than\ $t$\ participants which can reconstruct the secret by themselves, forming a privileged coalition[20].

\subsection{ Minimal privileged coalition}
\begin{definition}\label{2.3}{\cite{Spiez}} Let\ $r< t$ and fix $j$, $0< j < t-1.$ A coalition of $r$ participants $\textbf{L}=(l_{1},l_{2},\cdots,l_{r})\in F^{r}_{q}$ is said to be a $(t,j)-$privileged coalition, if they can reconstruct the secret, placed as the coefficient $a_{j}$ of the scheme polynomial $f(x)=a_{0}+a_{1}x+ \cdots +a_{t-1}x^{t-1}$.
\end{definition}
\begin{definition}\label{2.4}Let\ $r< t$ and fix $j$, $0< j < t-1.$ A coalition of $r$ participants $\textbf{L}=(l_{1},l_{2},\cdots,l_{r})\in F^{r}_{q}$ is said to be a $(t,j)-$minimal privileged coalition, if $\textbf{L}$ is a privileged coalition without any subtracks that can reconstruct the secret.
\end{definition}

Note that the tracks of length $t$, which can be extended by privileged coalitions of length $r$, can reconstruct the secrets placed as any coefficients of the scheme polynomial. Then we would have

\begin{definition} \label{2.5} Fix $j$, $0< j < t-1.$ A coalition of $t$ participants $\textbf{L}=(l_{1},l_{2},\cdots,l_{t})\in F^{t}_{q}$ is said to be a $(t,j)-$unextended track, if $\textbf{L}$ can not be extended by a $(t,j)-$privileged coalition.
\end{definition}

By Definition 2.3, the minimal authorized subsets of Shamir's secret sharing schemes are determined by minimal privileged coalitions and unextended tracks. Therefore, constructing minimal privileged coalitions become the key to the determination of the access structures. In order to devise an algorithm to obtain privileged coalitions, we will need the following theorem about the characterization of $(t,j)-$privileged coalitions whose proof can be found in[20].

\begin{theorem} \label{2.5} Assume that $0< j < t-1$, $j<r\leq t-1$.
Let $\textbf{L}=(l_{1},l_{2},\cdots,l_{r})\in F^{r}_{q}$ be a track, and let $t\leq q$. Then $\textbf{L}$ is a $(t,j)-$privileged coalition if and only if $$\tau_{\omega}(\textbf{L})=0, \ \ \ for \ all\  \omega\in \{r-j,\cdots,t-1-j\},$$
where $\tau_{\omega}(\textbf{L})$ denotes the elementary symmetric polynomial of total degree $\omega$.
\end{theorem}
\subsection{Privileged coalitions of any length}
\label{prel}In this section, we solve an open problem proposed in[20] recently, namely to design an algorithm of privileged coalitions of any length if such coalitions exist. For simplicity, we will focus on the tracks whose value of each component is clamped to the range of $\{1,2,\cdots,N\}$, and we call these coalitions $(t,j)-$privileged coalitions with respect to $N$. Using Algorithm 1 we can obtain $(t,j)-$privileged coalitions with respect to $N$ of length $r$ over $F_{p}$\ ($p$\ is a odd prime). By Definition 2.3, for the given $t,r$ with $\frac{t+1}{2}\leq r\leq t-1$, we can obtain $(t,j)-$minimal privileged coalitions with respect to $N$ of length $r$ by detaching ${r-r_{min} \choose p-r_{min}}N_{min}$ non-minimal privileged coalitions from $(t,j)-$privileged coalitions of length $r$, where $N_{min}$ denotes the number of privileged coalitions of the shortest length $r_{min}$.

As an illustration, we investigate the Shamir's secret scheme with the number of participants $n=13$ and threshold $t=7$. In the appendix, we then present two tables of $(7,j)-$minimal privileged coalitions with respect to $N=13$ of any length if such coalitions exist.

\textbf{\textbf{Algorithm 1}}

\textbf{Input} \quad positive integers\ $t,r,j$ with\ $t\geq 3$, $\frac{t+1}{2}\leq r\leq t-1$, $t-r\leq j\leq r-1$.

\textbf{Output} \quad all of the $(t,j)-$privileged coalitions of length $r$ with respect to $N$  over $F_{p}$.

 1.  Compute the range of values for $j$, and set $a\leftarrow r - j$, $b\leftarrow t - 1 - j$. Take

 \hspace{0.3cm} $J$ is a set of all integers from  $a$ to $b$.

2.  Obtain all of the tracks of length $r$ with respect to $N$  over $F_{p}$.

   \hspace{0.5cm}2.1. \quad $B=(1,2,\cdots,n)$, for \ $i=1$\  to\  $N$

   \hspace{1cm}2.1.1\quad For\  $j=i+1$\  to\  $N$

   \hspace{2cm}\quad If \ $B(j)>B(i)$  \ \ ($B(j)$ denotes the elements of the $j-$th position

   \hspace{5.6cm}of the array $B$),

   \hspace{1.5cm}2.1.1.1 \quad For \ $k=j+1$\  to\  $N$

   \hspace{2cm}\quad If \ $B(k)>B(j)$

   \hspace{2.5cm}\ \ \ \ \ \ \ \ \ \ \ \vdots

  \hspace{1.5cm}  \quad $r-1$ Nested loops are carried out in turn, then $\underbrace{(\cdots, B(k), B(j), B(i))}_{r}$

   \hspace{1.8cm} is a track.

  3.  Find $(t,j)-$privileged coalitions from the tracks obtained in Step 2 by means

    \hspace{0.5cm}of the principle about Vieta theorem of high power equation.

   \hspace{0.5cm}3.1.\quad For\  $i=1$ to\  $b-a+1$

  \hspace{1cm} 3.1.1.\quad  Set $p\leftarrow J(i)$, and go through tracks that obtained in Step 2.

  \hspace{2.5cm}Let $C=1$. For\  $j=1$\  to\ $r$

  \hspace{2.5cm}\quad Set $C\leftarrow C\times (x+\textbf{L}(j))$

   \hspace{1cm} 3.1.2.\quad Set$f\leftarrow C$, and select\ $\textbf{m}$ as a vector whose components are

   \hspace{2.5cm}coefficients of the expansion of $f(x)$ in ascending power of $x$.

  \hspace{2cm}\quad  Set $s\leftarrow \textbf{m}(r-p+1)$, then set $v_{p}\leftarrow mod (s,p)$.

  \hspace{0.5cm} 3.2.\quad If \ $v_{p}=0$, then return\ $(\textbf{L})$.

\section{An ideal multi-secret sharing scheme}
\label{intro}When the probability distributions over the secrets and shares are uniform, a secret sharing scheme is said to be ideal if all secrets and shares are the same size[19]. In this section we firstly define a (t-1)-tuple $\textbf{$\Gamma$}=(\Gamma_{0},\ldots,\Gamma_{t-2})$ of access structures and then we devise  an IMSSS which realizes such a (t-1)-tuple $\textbf{$\Gamma$}=(\Gamma_{0},\ldots,\Gamma_{t-2})$ of access structures.
\subsection{ Definition of the access structures}
Let\ $P=\{P_{1},\ldots,P_{n}\}$\ be the set of participants, we firstly define such an (t-1)-tuple  $\textbf{$\Gamma$}=(\Gamma_{0},\ldots,\Gamma_{t-2})$ as follows:

\hspace{0.03cm}(1)\quad \ $(\Gamma_{0})_{min}=\{A\subseteq P\mid |A|=t\}.$
\newpage
(2) \quad $(\Gamma_{j})_{min}=\{A\subseteq P \mid \textbf{L}_{A}\  is\  either\  a\  (t,j)-unextended\  track\  or\  a\  (t,j)-$

\hspace{4.6cm}$privileged\  coalition,(1\leq j\leq t-2)\}$.

As each component of a track  can be used as the identity of the participants, we can determine the minimal authorized subsets of  $(\Gamma_{j})_{min}$ for the secret placed as a coefficient $a_{j}$ of the scheme polynomial $f(x)=a_{0}+a_{1}x+\cdots+a_{t-1}x^{t-1}$ by getting both (t,j)-unextended track and (t,j)-privileged coalition using the  Algorithm 1.

Obviously, a subset $A\subseteq P$ is likely to reconstruct more than one secret. For example, if $A\in \Gamma_{i}$ and $A\in \Gamma_{j}$, then $A$ can reconstruct not only  $s_{i}$ but also $s_{j}$, where  $0\leq i,j\leq t-2$ and $i\neq j$.
\\
\\
\textbf{Example 3.1} \quad Let\ $P=\{P_{1},P_{2},P_{3},P_{4},P_{5},P_{6}\}, t=5,q=7.$ Without loss of generality, D distributes $i$ to participant $P_{i}$ for $1\leq i\leq 6.$ It follows from algorithm 1 that (5,1)-privileged coalition over $F_{7}$ is $(1,2,5,6),\ (1,3,4,6)$\ and\ $(2,3,4,5)$;\ (5,2)-privileged coalition is $(1,2,4),(3,5,6)$;\ (5,3)-privileged coalition is \ $(1,2,5,6),\ (1,3,4,6)$\ and\ $(2,3,4,5)$, but without (t,j)-unextended track. Hence  $(t-1)$-tuple  $\textbf{$\Gamma$}=(\Gamma_{0},\ldots,\Gamma_{m-2})$ can be defined as follows:

$(\Gamma_{0})_{min}=\{\{P_{1},P_{2},P_{3},P_{4},P_{5}\},\{P_{1},P_{2},P_{3},P_{4},P_{6}\}\{P_{1},P_{2},P_{4},P_{5},P_{6}\},$

\ \ \ \ \ \ \ \  \ \ \   \ \ \ $\{P_{1},P_{2},P_{3},P_{5},P_{6}\},\{P_{1},P_{3},P_{4},P_{5},P_{6}\},\{P_{2},P_{3},P_{4},P_{5},P_{6}\}\},$

$(\Gamma_{1})_{min}=(\Gamma_{3})_{min}=\{\{P_{1},P_{3},P_{4},P_{6}\},\{P_{1},P_{2},P_{5},P_{6}\},\{P_{2},P_{3},P_{4},P_{5}\}\},$
\ \ \ \ \

$(\Gamma_{2})_{min}=\{\{P_{1},P_{2},P_{4}\},\{P_{3},P_{5},P_{6}\}\}.$

\subsection{ Construction of the IMSSS}
\subsubsection{Initialization phase}

\label{intro}Note that $s_{0},\ldots,s_{t-2}$ denote t-1 secrets to be shared, where $(s_{0},\ldots,s_{t-2})\in S_{0}\times\cdots\times S_{t-2}.$ The dealer D randomly chooses $n$ pairwise different $l_{i}$ and  assigns them to every participant $P_{i}$ as their public identity, where $l_{i}\in F_{q}^{*}$ for $1\leq i\leq n.$ Then D computes the $(t,j)$-unextended track and $(t,j)$-privileged coalition by means of algorithm 1 for $0\leq j\leq t-2.$
\subsubsection{Distribute phase}
\hspace{0.2cm}The dealer D performs the following steps:

\hspace{-0.5cm}(1) \ \ Choose an integer $a_{t-1}$ from $F_{q}^{*}$ and construct $(t-1)$th degree polynomial

\hspace{0.45cm}$f(x)$ mod $q$, where $0<s_{0},\ldots,s_{t-2},a_{t-1}<q$ as follows:$f(x)=s_{0}+\cdots+$

\hspace{0.45cm}$s_{t-2}x^{t-2}+a_{t-1}x^{t-1}\mod q.$

\hspace{-0.5cm}(2) \ \ Compute $y_{i}=f(l_{i})\mod q $ for $i=1,2\ldots,n$ and  distribute them to every

\hspace{0.45cm}participant $P_{i}$ as secret shares by a secret channel.
\subsubsection{Recovery  phase}
Assume that $(0\leq j\leq t-2).$ for any $A\in (\Gamma_{j})_{min}$, if $|A|=t$, then any subset $A\in (\Gamma_{j})_{min}$ can reconstruct the secret $s_{j}$ by solving equation system $\textbf{A}(\textbf{L}_{A})\ \textbf{a}^{T}=\textbf{y}^{T}$, where $\textbf{L}_{A}$ is corresponding track of $A$; if $|A|<t$ and $\textbf{L}=(l_{i_{1}},l_{i_{2}},\ldots,l_{i_{r}})\in F^{r}_{q}$ is a $(t,j)$-privileged coalition of $A$, let $\textbf{L}^{'}_{A}=(l_{i_{1}},l_{i_{2}},\ldots,l_{i_{r}},l_{i_{r+1}},\ldots,l_{i_{t}})\in F^{t}_{q}$, then the participants in $A$ can reconstruct the secret $s_{j}$ by solving equation system $\textbf{A}(\textbf{L}^{'}_{A})\ \textbf{a}^{T}=\textbf{y}^{T},$ where $\textbf{a}=(s_{0},\ldots, s_{t-2}, a_{t-1})\in F^{t}_{q},$  $\textbf{y}=(y_{i_{1}},\ldots,y_{i_{t}})$,
and $\textbf{A}(\textbf{L}_{A})$,
$\textbf{A}(\textbf{L}^{'}_{A})$ are all $t\times t$ matrix which can be specifically written as $(l^{\nu}_{\mu})_{1\leq \mu\leq n
\atop {0\leq \nu\leq t-1}}.$
\section{Correctness and security proof}
In order to prove that our scheme is perfect, we need to introduce some results on the generalized the Vandermonde determinants. As usual, for a $k-$tuple of indeterminates $\textbf{x}=(x_{1},\ldots,x_{k})$ and a $k-$tuple of increasing non-negative integers $\textbf{c}=(c_{1},\ldots,c_{k})$ we call $V_{\textbf{c}}(\textbf{x})=\textrm{det}((x^{c_{\nu}}_{\mu})_{1\leq \nu,\mu\leq k})$ generalized Vandermonde determinant. Write $\textbf{e}_{k}=(0,\ldots,k-1).$ If $\textbf{c}=\textbf{e}_{k}$ then $V_{\textbf{c}}(\textbf{x})$ equals the classical Vandermonde determinant
$V(\textbf{x})=\prod_{1\leq i<j\leq k}(x_{j}-x_{i})$.

\begin{theorem} \label{4.1}The scheme presented in Section 3 is a perfect multi-secret sharing scheme.
\end{theorem}
\begin{proof} When\ $j=0$\, due to the perfect property of the $(t,n)$ secret sharing schemes, our scheme satisfies the two conditions of Definition 1.1. Now we consider that $1\leq j\leq t-2.$

(1) \quad  If $|A|=t$, then by solving equation system $\textbf{A}(\textbf{L}_{A})\ \textbf{a}^{T}=\textbf{y}^{T},$ the participants in $A$  can obtain the unique solution ${s_{j}}$ in terms of Cramer rule.  If $|A|<t$ and $\textbf{L}=(l_{i_{1}},l_{i_{2}},\ldots,l_{i_{r}})\in F^{r}_{q}$ is a $(t,j)$-privileged coalition of $A$, by Theorem 2.5, $\tau_{\omega}(\textbf{L})=0,$ for all $\omega\in \{r-j,\cdots,t-1-j\}$. By Lemma 2[20], $\tau_{\omega}(\textbf{L})=0$ for all $\omega\in \{r-j,\cdots,t-1-j\}$ if and only if
 $$\quad\quad\quad\quad\quad\quad\tau_{t-1-j}(\textbf{L}\ ||\ \hat{\textbf{u}}_{m})=0 \ \ \ for\  all \ m,\  1\leq m\leq t-r, \quad\quad\quad\quad      (1^{'})$$
let $\textbf{u}=(u_{1},\ldots,u_{t-r})\in F_{q}^{t-r}$ be a track disjoint with $\textbf{L}$, and $\hat{\textbf{u}}_{m}$\ denotes the sequence obtained from $\textbf{u}$ by removing the term $u_{m}.$ Let\begin{displaymath}
y_{k} = \left\{ \begin{array}{ll}
f(l_{k}) & \textrm{µ±\ $k \in \{1,\ldots,r\}$}\\
f(u_{k-r})& \textrm{µ±\ $k \in \{r+1,\ldots,t\}$}\\
\end{array} ,\right.
\end{displaymath}
and let $\textbf{L}^{'}_{A}=(l_{i_{1}},l_{i_{2}},\ldots,l_{i_{r}},u_{1},\ldots,u_{t-r})\in F^{t}_{q}$, then $s_{j}$ can be obtained by solving  equation system
$\textbf{A}(\textbf{L}^{'}_{A})\ \textbf{a}^{T}=\textbf{y}^{T}$
\begin{eqnarray*}
 s_{j}&=&\frac{1}{V(\textbf{L}\ ||\ \textbf{u})}(\sum^{r}_{k=1}(-1)^{k+j+1}V(\hat{\textbf{L}}_{k}\ ||\ \textbf{u})
\tau_{t-1-j}(\hat{L}_{k}\ ||\ \textbf{u})y_{k}+\\& &
\sum^{t}_{k=r+1}(-1)^{k+j+1}V(\textbf{L}\ ||\ \hat{\textbf{u}}_{k-r})
\tau_{t-1-j}(\textbf{L}\ ||\ \hat{\textbf{u}}_{k-r})y_{k}.\quad (2^{'})
\end{eqnarray*}
By $(1^{'})$ and $(2^{'}),$ $(t,j)$-privileged coalition of $A$ can compute the secret $s_{j}$. Hence, it holds that for all $A\in \Gamma_{j}, H(S_{j}\ |\ A)=0.$

(2) If $ A\not\in \Gamma_{j}$, then $|A|<t$ and $\textbf{L}_{A}$ is not a $(t,j)-$ privileged coalition. In view of $(1^{'})$, there exists a $m^{'}$, where $1\leq m^{'}\leq t-r$ such that $\tau_{t-1-j}(\textbf{L}\ ||\ \hat{\textbf{u}}_{m^{'}})\neq 0$. Thus, the share $y_{m^{'}+r}$ of $u_{m^{'}}$ is needed. By $(2^{'}$) we can obtain that the participants in $A$ have no information on $s_{j}$, even knowing some of the other secrets. Hence, it holds that $ H(S_{j}\ |\ AT)=H(S_{j}\ |\ T),$ where $T$ denotes the secrets that $A$ can compute.

Therefore, according to Definition 1.1, the scheme is a perfect multi-secret sharing scheme.
\end{proof}

As a consequence, our scheme is an ideal and perfect linear multi-secret sharing scheme. In 2004 and 2005, Yang et al.[13] and Pang et al.[14] proposed multi-secret sharing schemes based on Shamir's threshold secret sharing, respectively, which are relatively efficient with lower cost of computing due to the Lagrange interpolation operation that is employed
in the process of schemes construction. However, the existence of the privileged coalitions and the Lagrange interpolation polynomial  participants use will result in a fact that the union of less than $t$ participants may compute the coefficients of the polynomial corresponding to secrets by solving equation system, thereby obtaining further information about the secret. Consequently, both of the two schemes are not perfect, i.e., there is information leakages. Table 1 is for the comparison among three schemes.

Likewise, the multi-secret sharing schemes[15-18] are not perfect, and thus are not ideal.
\\
\textbf{Remark 2.}  The validity of the shares can be verified in a verifiable secret sharing scheme, thus participants are not able to
cheat. Based on our scheme, we can further construct an ideal verifiable multi-secret sharing scheme by adding the existing
verifiability methods where the intractability of discrete logarithm problem is frequently employed (see[15-18]).

{\vspace{0.5cm}
{\scriptsize{\textbf{Table\ 1} \ \ The comparison of performance among three schemes
}\vspace{-4mm}\\
\begin{center} \doublerulesep 1.5pt \tabcolsep 6pt
\begin{tabular*}{\textwidth}{lccc}
\hline Capability &Our scheme & Yang's scheme. & Pang's scheme \\ \hline
 Multi-secret  & Yes&Yes&Yes\\  Each participant holds only one share &  Yes&Yes&Yes\\   Recover multi-secrets by Lagrange interpolating polynomials&  No&Yes&Yes\\  Access structures corresponding to each secret is the same & No&Yes&Yes\\  Access structures possess more vivid authorized sets&  Yes&No&No\\  The scheme is perfect &  Yes&No&No\\ The scheme is ideal &  Yes&No&No\\ \hline
\end{tabular*}
\end{center}}

\section{Conclusions}
In this paper, we consider an ideal multi-secret sharing scheme based on the theories of minimal privileged coalitions and Shamir's secret sharing, where for a set of participants $P=\{P_{1},\ldots,P_{n}\}$, each subset of $\Gamma_{j}$ carries different target secret $s_{j}$ for $0\leq j\leq t-2$. In particular, in order to obtain privileged coalitions, we devise an algorithm of privileged coalitions of any length if such coalitions exist. In real terms, we can integrate data into a database which obtained from the experiment for different value $t$ in accordance with our needs, and then we can extract the results as identities of participants applied to the construction of the scheme for practical applications.

\section{Acknowledgements}

This work was supported by the National Natural Science Foundation of China (Grant No. 60873119).



\bigskip

\newpage

\begin{center} \textbf{Table\ 2} \ \ The number of $(7,j)$-minimal privileged coalitions in the track $(1,\ldots,13)$\ over\ $F_{p}$\ for \ $1\leq j\leq5$ \end{center}\vspace{-4mm}
\begin{center} \doublerulesep 0.2pt \tabcolsep 7pt
\begin{tabular*}{\textwidth}{c c c c c c r c c c c c}\hline
 p  &   j=1 &   j=2 &   j=3	&  j=4  &   j=5	&     p &      j=1 &    j=2 &    j=3&    j=4&    j=5 \\ \hline
13	&	72	&	114	&	71	&	93	&	132	&¡¡¡¡137&		15&		10 &		 11&		18	&	0\\
17  &	101	&	69	&	76  &	72	&	98	&¡¡¡¡139&		18&		12 &		 10&		13&		0\\
19	&	90	&	53	&	72	&	57  &	92	&¡¡¡¡149&		13&		13 &		 8&		8&		0\\
23	&	72  &   61	&	58  &	63  &	83	&¡¡¡¡151&		12&		9 &		    13&	    13&	    0\\
29	&	60	&	43	&	51	&	46  &	26	&¡¡  157&       14&     10 &        13&     10 &    0\\
31	&	58	&	51	&	40	&	39	&	32	&¡¡¡¡163&		13	&	16	&	 9	&	7	&	0\\
37	&	51	&	49	&	38	&	27	&	76	&¡¡  167&		13	&	12	&	 9	&	11	&	0\\
41  &	44	&	36	&	38	&	27	&	94	&¡¡¡¡173&		8	&	13	&	 11	&	10	&	0\\
43	&	41	&	40	&	35	&	34	&	94	&¡¡¡¡179&		7	&	11	&	 9	&	10	&	0\\
47  &	41	&	32	&	36	&	39	&	76	&¡¡¡¡181&		7	&	8	&	 12	&	9	&	0\\
53	&	26	&	35	&	35	&	27	&	31	&¡¡¡¡191&		5	&	4	&	 11	&	10	&	0\\
59	&	32	&	25	&	28	&	24	&	5	&¡¡  193&		12	&	8	&	 7	&	8	&	0\\
61  &   27	&	32	&	31	&	28	&	2	&¡¡  197&		10	&	10	&	 11	&	9	&	0\\
67  &	21	&	32	&	25	&	22	&	0	&¡¡  199&		7	&	13	&	 9	&	8	&	0\\
71	&	26	&	22	&	22	&	23	&	0	&  \vdots & \vdots &  \vdots & \vdots & \vdots & \vdots \\
73	&	23	&	24	&	21	&	31	&	0	&¡¡  809 &		1	&	2	&	 3	&	0	&	0\\
79	&	24	&	20	&	22	&	15	&	0	&¡¡ \vdots &  \vdots & \vdots & \vdots & \vdots & \vdots \\
83	&	23	&	21	&	24	&	25	&	0	&¡¡  5231 &	    0   &	0   &	 1   &	0   &   0\\
89	&	18	&	12	&	15	&	18	&	0	&¡¡$\vdots$	 &$\vdots$&$\vdots$&$\vdots$&$\vdots$&$\vdots$\\
97	&	18	&	24	&	15	&	17	&	0	&¡¡ 31601&	    0   &	1	&	 0	&	0	&	0\\
101	&	19	&	10	&	21	&	17	&	0	&¡¡$\vdots$	 &$\vdots$&$\vdots$&$\vdots$&$\vdots$&$\vdots$\\
103	&	18	&	20	&	14	&	20	&	0	&¡¡ 199999  &	0&		0	&	 0	&	0&		0\\
107	&	21	&	14	&	12	&	18	&	0	&¡¡$\vdots$	 &$\vdots$&$\vdots$&$\vdots$&$\vdots$&$\vdots$\\
109	&	18	&	18	&	9	&	16	&	0	&¡¡ 499253  &	0	&	0	&	 0	&	0&		0\\
113	&	12	&	14	&	16	&	16	&	0	&¡¡$\vdots$	 &$\vdots$&$\vdots$&$\vdots$&$\vdots$&$\vdots$\\
127	&	12	&	12	&	16	&	19	&	0	&  $\vdots$	 &$\vdots$&$\vdots$&$\vdots$&$\vdots$&$\vdots$\\\
131	&	13	&	7	&	11	&	12	&	0	&    $\geq725597$	&	0	&	 0	&	0	&	0	&	0        \\ \hline
\end{tabular*}
\end{center}}\vspace{2mm}

\newpage
\begin{center}
\scriptsize{\textbf{Table\ 3} \ \ The  \ $(7,3)$-privileged coalitions with the shortest length in the track $(1,\ldots,13)$\ over\ $F_{p}$}
\end{center}
\vspace{-0.5cm}
\scriptsize{$$\begin{tabular*}{\textwidth}{c| c| lp{15cm}}\hline
p & Nr. & $(7,3)$-privileged coalitions with the shortest length\\ \hline
13&	3&	\{12, 8, 5, 1\}  \{11, 10, 3, 2\}  \{9, 7, 6, 4\}\\ \hline
17&	1&	\{11, 10, 7, 6\}\\ \hline
19&	3&	\{13, 5, 4, 3, 1\} \{13, 9, 8, 5, 2\} \{10, 8, 7, 6, 2\}\\ \hline
23&	2&	\{10, 9, 8, 6, 1\} \{13, 11, 8, 7, 1\}\\ \hline
29&	3&	\{11, 4, 3, 2, 1\} \{11, 9, 7, 5, 1\} \{12, 9, 6, 4, 3\}\\ \hline
31&	1&	\{9, 8, 5, 3, 1\}\\ \hline
37&	2&	\{12, 10, 8, 5, 1\} \{7, 6, 4, 3, 2\}\\ \hline
41&	1&	\{11, 9, 7, 6, 4\}\\ \hline
&  & \{10, 5, 4, 3, 2, 1\}\{12, 6, 5, 3, 2, 1\} \{11, 10, 9, 5, 3, 1\} \{13, 7, 5, 4, 2, 1\}\{12, 11, 6, 4, 2, 1\}\\ & &\{13, 11, 9, 5, 2, 1\} \{11, 10, 7, 6, 2, 1\}\{11, 8, 6, 3, 2, 1\}
\{12, 10, 8, 6, 5, 4\}\{13, 12, 11, 6, 5, 4\}\\ 43& 35& \{11, 9, 8, 5, 4, 1\}\{13, 9, 8, 6, 4, 1\}\{11, 10, 8, 7, 4, 1\}\{12, 11, 10, 9, 4, 1\}\{12, 11, 8, 6, 5, 1\}\\ & &\{13, 11, 10, 8, 6, 1\}\{13, 9, 5, 4, 3, 2\}\{12, 10, 8, 6, 3, 2\}
\{10, 9, 8, 7, 3, 2\}\{13, 12, 10, 5, 4, 2\}\\ & &\{12, 10, 8, 7, 4, 2\}\{12, 9, 8, 6, 5, 2\}
\{13, 10, 8, 6, 5, 2, \}\{11, 9, 8, 7, 6, 2,\}\{8, 7, 5, 4, 3, 1\} \\ & &\{12, 11, 10, 6, 5, 3\}\{13, 12, 11, 8, 6, 3\}\{11, 10, 9, 8, 7, 5\}\{11, 7, 6, 5, 4, 3\}\{12, 11, 8, 4, 3, 2\} \\ &&\{11, 9, 7, 4, 2, 1\}\{13, 10, 8, 7, 6, 4\}\{13, 11, 9, 7, 6, 1\}\{11, 10, 9, 8, 4, 3\}\{13, 11, 9, 6, 4, 2, \}\\ \hline
47&	1&	\{13,  11,  9,  8,  2\} \\ \hline
&	& \{11, 6, 5, 3, 2, 1\} \{12, 10, 7, 3, 2, 1\}\{7, 6, 5, 4, 2, 1\}\{12, 11, 8, 4, 2, 1\}\{12, 8, 6, 5, 2, 1\}\\ & & \{11, 9, 7, 5, 2, 1\}\{11, 10, 9, 6, 2, 1\}
\{12, 11, 5, 4,  3,  1\}\{12, 9, 7, 5, 3, 1\}\{13, 10, 8, 6, 3, 1\}\\ & & \{13, 12, 11, 10, 3, 1\}\{9, 8, 7, 6, 4, 1\}\{13, 10, 9,  6,  4,  1\}\{13, 12, 10, 6, 5, 1\}
\{13, 12, 9,  8,  6,  1\}\\ 53& 35&\{9,  8,  6,  5,  3,  2\}\{13,  9,  8,  7,  3,  2\}\{13,  11,  8,  7,  4,  2\} \{13,  10,  8,  7,  5,  3\}\{13,  11,  8,  7,  6,  3\} \\ & &\{10, 8, 7, 5, 4, 2\}\{13, 12, 7, 5, 4, 2\}\{11,  10,  9,  5,  4,  2\}\{13,  11,  10,  9,  8,  2\} \{12,  11,  9,  7,  6,  2\} \\& &
\{11,  8,  6,  5,  4,  3\}\{12,  10,  9,  5,  4,  3\}\{11,  10,  8,  7,  4,  3\}\{11,  9,  7,  6,  5,  3\}\{13, 12, 11, 9, 2, 1\}\\&&\{3, 9, 6, 5, 2, 1\} \{10, 9, 8, 7, 3, 1\}\{13, 12, 10, 4,  3,  2\}\{12,  11,  10,  7,  6,  3\}\{13,  12,  10,  8,  6,  2\}\\ \hline
&	&	\{10, 7, 6, 3, 2, 1\}\{11 8, 5, 4, 2, 1\}\{12, 10, 5, 4, 2, 1\} \{10,  8,  7,  4,  2,  1\}\{11,  7,  6,  5,  2,  1\}\\ & & \{12,  9,  8,  7,  6,  5\}
\{12,  9,  8,  6,  3,  1\}\{12,  8,  7,  5,  4,  1\}\{13,  10,  9,  7,  4,  1\}\{13,  12,  8,  6,  5,  1\}\\ 59& 28& \{11,  10,  7,  4,  3,  2\}\{13,  12,  11, 10,  6, 5\}
\{13,  11,  8,  4,  3,  2\}\{13,  11,  10,  6,  3,  2\}\{13,  12,  10,  8,  3,  2\}\\ & & \{13,  11,  10,  8,  5,  2\}
\{13,  12,  11,  9,  8,  2\}\{12,  9,  7,  5,  4,  3\}\{13,  12,  11,  7,  4,  3\}\{11,  10,  9,  8,  4,  3\}\\ & & \{12, 11, 10, 9, 8, 7\}\{11,  10,  8,  6,  5,  4\}\{13, 8, 7, 6, 3, 1\}\{10,  9,  8,  7,  5,  3\}\{13,  6,  5,  4,  3,  1\}\\&&\{12,  10,  9,  5,  3,  2\}\{13,  10,  7,  6,  5,  2\}\{12,  11,  9,  8,  5,  1\}\\  \hline
&	 & \{10,  9,  5,  3,  2,  1\}\{11,  10,  6,  3,  2,  1\}\{12,  10,  7,  4,  2,  1\}\{12,  9,  6,  5,  2,  1\}\{13,  11,  7,  5,  2,  1\}\\ & &
\{13, 8, 6, 4, 3, 1\}\{12, 11, 6, 5, 3, 1\}\{11, 9, 8, 5, 3, 1\}\{12, 11, 10, 5, 4, 1\}\{10,  8,  7,  6,  5,  1\}\\ & &
\{13,  12,  11,  8,  7,  1\}\{12, 6, 5, 4, 3, 2\}\{10,  9,  6,  4,  3,  2\}\{12,  10,  8,  5,  3,  2\}\{13, 12,  10,  7,  6,  1\} \\ 61& 31&
\{13,  11,  10,  5,  3,  2\}\{13,  9,  7,  6,  3,  2\}\{12,  11,  8,  7,  3,  2\}\{13,  12,  9,  7,  4,  2\}\{10,  9,  7,  6,  5,  2\}\\
& &
\{12,  10,  7,  5,  4,  3\}\{13,  11,  9,  6,  4,  3\}\{13,  12,  10,  8,  4,  3\}\{13,  11,  9,  4,  2,  1\}\{13, 12, 10, 9, 5, 1\}\\&&\{11,  10,  7,  6,  5,  3\} \{13,  12,  8,  6,  2,  1\}\{13, 12, 11, 9, 6, 5\} \{13,  12,  8,  6,  5,  4\}\{13,  12,  11,  10,  9,  2\}\\ &&\{13,  12,  9,  7,  5,  3\}\\ \hline
&	& \{12, 10, 6, 3, 2, 1\}\{10,  8,  7,  3,  2,  1\}\{8,  6,  5,  4,  2,  1\}\{11,  9,  8,  5,  2,  1\}\{12,  11,  9,  6,  2,  1\}\\& &
\{13,  12,  11,  6,  4,  1\}\{10,  8,  7,  6,  5,  1\}\{12,  9,  8,  7,  6,  1\}\{12,  11,  10,  8,  6,  1\}\{9,  7,  6,  5,  3,  2\}
\\ 67& 25& \{12,  9,  6,  5,  4,  2\}\{13,  9,  8,  6,  5,  2\}\{13,  9,  7,  5,  4,  3\}\{12,  11,  8,  5,  4,  3\}\{9,  8,  7,  6,  4,  3\}
\\ & & \{12,  11,  9,  6,  5,  3\}\{13,  11,  10,  8,  7,  3\}\{13,  10,  8,  7,  5,  4\}\{13,  11,  9,  8,  5,  4\}\{11, 9, 7, 5, 4, 1\}\\&&\{13,  12,  10,  9,  8,  6\}\{13, 12,  8, 5, 3, 1\}\{13,  12,  9,  5,  3,  2\}\{11,  10,  7,  6,  4,  3\}\{12,  11,  9,  8,  4,  1\}\\ \hline
&		&\{13,  12,  4,  3,  2,  1\}\{13,  10,  9,  3,  2,  1\}\{11,  10,  6,  4,  2,  1\}\{7,  6,  5,  4,  3,  1\}\{13,  8,  5,  4,  3,  1\}
\\ & & \{9,  8,  5,  4,  2,  1\}\{11,  9,  8,  4,  3,  1\}\{12,  11,  9,  7,  4,  1\}\{10,  9,  7,  6,  5,  1\}\{12,  11,  10,  9,  6,  1\}\\ 71& 22&
\{10,  9,  7,  4,  3,  2\}\{12,  10,  8,  5,  4,  2\}\{13,  12,  11,  6,  4,  2\}\{13,  11,  10,  7,  4,  2\}\{13,  12,  9,  8,  5,  2\}\\ & &
\{10,  9,  8,  6,  4,  3\}\{12,  11,  8,  7,  4,  3\}\{13,  12,  10,  9,  7,  3\}\{13,  12,  11,  10,  9,  5\}\{13,  9,  7,  4,  3,  1\}\\&&\{10,  6,  5,  4,  3,  2\}\{13,  9,  8,  7,  6,  2\}\\ \hline
&	&	\{13,  8,  4,  3,  2,  1\}\{10,  8,  5,  3,  2,  1\}\{12,  10,  9,  5,  2,  1\}\{12,  11,  9,  6,  2,  1\}\{10,  9,  7,  4,  3,  1\}\\ & &
\{13,  10,  7,  5,  3,  1\}\{13,  9,  8,  7,  3,  1\}\{12,  11,  8,  7,  6,  1\}\{12,  11,  10,  8,  3,  2\}\{11,  9,  8,  5,  4,  2\}\\ 73& 21&
\{13, 11, 7, 6, 5, 2\}\{11,  10,  9,  6,  5,  2\}\{13,  12,  8,  7,  5,  2\}\{13,  11,  10,  9,  7,  2\}\{11,  10,  8,  6,  4,  3\}\\ & & \{12,  10,  9,  8,  6,  4\}\{13,  12,  9,  8,  5,  4\}\{12,  11,  10,  7,  5,  3\}\{13, 11, 10, 8, 4, 2\}\{13,  12,  9,  4,  3,  1\}\\&& \{12,  8,  7,  6,  5,  3\}\\ \hline
&	&	\{9,  7,  6,  3,  2,  1\}\{13,  11,  10,  3,  2,  1\}\{13,  10,  5,  4,  2,  1\}\{13,  7,  6,  5,  2,  1\}\{10,  9,  6,  5,  2,  1\}\\ & &
\{12,  11,  7,  5,  3,  1\}\{10,  9,  7,  5,  4,  1\}\{12,  11,  8,  5,  4,  1\}\{9,  6,  5,  4,  3,  2\}\{11,  9,  7,  5,  3,  2\}\\ & &
\{13,  10,  8,  6,  5,  2\}\{11,  8,  6,  5,  4,  3\}\{13,  12,  10,  5,  4,  3\}\{13,  9,  8,  7,  4,  3\}\{12,  8,  7,  6,  5,  3\}\\ 79& 22&
\{13,  11,  9,  8,  5,  4\}\{12,  11,  10,  7,  6,  4\}\{13,  11,  10,  9,  6,  5\}\{13,  11,  10,  8,  7,  6\}\{13,  12,  10,  8,  2,  1\}\\&&\{12,  8,  7,  6,  4,  2\}\{13,  12,  9,  7,  5,  3\}\\ \hline
&	&\{12,  5,  4,  3,  2,  1\}\{13,  7,  4,  3,  2,  1\}\{8,  7,  6,  4,  2,  1\}\{12,  11,  7,  4,  2,  1\}\{13,  10,  5,  4,  3,  1\}\\ & & \{13,  10,  7,  6,  3,  1\}\{13,  12,  9,  8,  4,  1\}\{12,  9,  8,  6,  5,  1\}\{11,  10,  9,  8,  5,  1\}\{13,  11,  8,  7,  6,  1\}\\ 83& 24&
\{12,  9,  8,  6,  3,  2\}\{11,  10,  9,  6,  3,  2\}\{13,  9,  7,  6,  4,  2\}\{12,  10,  9,  6,  4,  2\}\{12,  11,  10,  8,  4,  2\}\\ & &
\{13,  12,  10,  7,  5,  2\}\{9,  8,  7,  5,  4,  3\}\{13,  12,  10,  8,  4,  3\}\{12,  10,  9,  7,  5,  4\}\{11,  10,  9,  8,  7,  4\}\\&&\{10,  8,  6,  5,  3,  1\}\{12,  11,  7,  5,  3,  2\}\{12,  11,  10,  6,  5,  2\} \{13,  12,  11,  10,  9,  7\}\\ \hline

\end{tabular*}
$$}

\scriptsize{$$\begin{tabular*}{\textwidth}{c| c |lp{5cm}}\hline
&	&	\{12, 9, 5, 4, 2, 1\}\{9,  8,  6,  5,  2,  1\}\{13, 11, 5, 4, 3, 1\}\{12,  11, 8, 5, 3, 1\}\{11,  8,  6,  5,  4,  1\}\\ 89& 15& \{13, 7, 5, 4, 3, 2\}\{13, 11, 10, 8, 7, 3\}
\{13,  11,  10,  4,  3,  2\}\{12,  10,  9,  5,  3,  2\}\{11,  10,  9,  6,  5,  2\}\\ & & \{12,  8,  6,  5,  4,  3\}\{13, 12, 7, 6, 5, 3\}\{13,  12,  11,  10,  6, 3\}\{13,  12,  8,  5,  4,  1\}\{12,  11,  10,  7,  6,  2\}\\ \hline
97&	1&	\{12,  11,  10,  5,  4\}\\ \hline
 &	&\{12,  10,  4,  3,  2,  1\}\{11,  10,  7,  3,  2,  1\}\{13,  9,  8,  3,  2,  1\}\{12,  10,  7,  5,  2,  1\}\{13,  10,  7,  6,  2,  1\}\\ & &\{13,  11,  9,  8,  6,  3\}
\{13,  12,  11,  7,  2,  1\}\{10,  7,  5,  4,  3,  1\}\{13,  11,  8,  4,  3,  1\}\{9,  7,  6,  5,  4,  1\}\\ 101& 21&\{11,  10,  9,  7,  6,  1\}\{13,  11,  7,  6,  5,  4\}
\{12,  11,  8,  7,  3,  2\}\{13,  9,  6,  5,  4,  2\}\{13,  10,  9,  8,  6,  2\}\\ & &\{12, 11, 10, 8, 4, 3\}\{12,  10,  9,  6,  5,  3\}\{13,  11,  10,  9,  8,  5\}\{11,  9,  8,  7,  2,  1\}\{12,  10,  9,  6,  4,  1\}\\ && \{13,  10,  8,  6,  4,  3\}\\ \hline
 &	&\{10, 7, 5, 4, 2, 1\}\{12, 10, 6, 4, 2, 1\}\{12, 11, 9, 4, 2, 1\}\{12,  11,  8,  6,  2,  1\}\{12,  7,  6,  4,  3,  1\}\\ 103& 14&\{12, 11, 10, 7,  5,  4\}
\{13, 10, 9, 5, 3, 2\}\{11, 10, 9, 5, 4, 2\}\{13, 12, 11, 9, 5, 2\}\{10, 9, 6, 5, 4, 3\}\\ & &\{12,  9,  7,  6,  5,  3\}\{13, 11, 10, 8, 7, 4\}\{9,  8,  5,  4,  3,  2\}\{13, 12,  10,  5,  4,  3\}\\ \hline
&		&\{13,  12,  10,  3,  2,  1\}\{10,  8,  6,  4,  2,  1\}\{12,  9,  8,  4,  2,  1\}\{12,  11,  7,  6,  2,  1\}\{13,  12,  7,  5,  4,  1\}\\ 107& 12&
\{11,  10,  9,  8,  6,  1\}\{11,  9,  8,  5,  3,  2\}\{13,  10,  7,  6,  5,  2\}\{13,  12,  11,  9,  8,  3\}\{12,  10,  9,  7,  5,  4\}\\&&\{13,  12,  11,  10,  4,  1\}\{13,  9,  8,  7,  6,  5\}\\ \hline
 &		&\{13,  12,  10,  3,  2,  1\}\{10,  8,  6,  4,  2,  1\}\{12,  9,  8,  4,  2,  1\}\{12,  11,  7,  6,  2,  1\}\{13,  12,  7,  5,  4,  1\}\\ 107& 12&
\{11,  10,  9,  8,  6,  1\}\{11,  9,  8,  5,  3,  2\}\{13,  10,  7,  6,  5,  2\}\{13,  12,  11,  9,  8,  3\}\{12,  10,  9,  7,  5,  4\}\\&&\{13,  12,  11,  10,  4,  1\}\{13,  9,  8,  7,  6,  5\}\\ \hline
 &	&	\{10,  9,  5,  4,  3,  1\}\{13,  12,  7,  4,  3,  1\}\{11,  10,  9,  7,  5,  1\}\{13,  11,  10,  8,  7,  1\}\{13,  12,  9,  7,  3,  2\}\\ 109& 9&
\{10,  8,  6,  5,  4,  3\}\{13,  12,  9,  8,  6,  4\}\{10,  9,  8,  7,  6,  5\}\{13,  7,  6,  5,  4,  3\}\\ \hline
113&	1&	\{13,  12,  7,  5,  2\}\\ \hline
&&\{9,  6,  5,  3,  2,  1\}\{13,  11,  9,  3,  2,  1\}\{13,  12,  5,  4,  2,  1\}\{11,  10,  8,  5,  3,  1\}\{13,  11,  8,  7,  3,  1\}\\ 127& 16&
\{13,  10,  9,  8,  7,  1\}\{13,  12,  8,  4,  3,  2\}\{13,  9,  8,  5,  3,  2\}\{12,  11,  8,  5,  3,  2\}\{11,  10,  8,  5,  4,  2\}\\ & &
\{10,  9,  8,  5,  4,  3\}\{13,  8,  7,  6,  5,  3\}\{13,  11,  10,  9,  5,  4\}\{13,  12,  10,  8,  6,  4\}\{9,  8,  7,  6,  5,  1\}\\&&\{13,  12,  10,  9,  7,  2\}\\ \hline
&		&\{11,  6,  5,  4,  2,  1\}\{13,  8,  6,  4,  2,  1\}\{12,  11,  9,  5,  2,  1\}\{9,  8,  6,  4,  3,  2\}\{13,  12,  11,  10,  8,  2\}\\ 131& 11&\{10,  8,  7,  6,  5,  3\}
\{13,  11,  10,  7,  5,  3\}\{12,  11,  10,  9,  5,  3\}\{13,  10,  8,  7,  5,  4\}\{13,  11,  10,  8,  6,  4\}\\&&\{13,  12,  8,  7,  6,  5\}\\ \hline
&	&	\{13,  12,  10,  7,  2,  1\}\{13,  7,  6,  4,  3,  1\}\{11,  10,  8,  5,  4,  1\}\{13,  11,  10,  9,  4,  1\{10,  8,  7,  5,  3,  2\}\\ 137&11 &
\{12,  11,  10,  9,  7,  2\}\{10,  9,  7,  6,  5,  3\}\{13,  11,  7,  6,  5,  3\}\{13,  9,  8,  7,  6,  4\}\{12,  11,  9,  8,  6,  5\}\\&&\{11,  10,  9,  8,  5,  2\}\\ \hline
&	&\{13,  12,  8,  3,  2,  1\}\{9,  7,  6,  4,  2,  1\}\{13,  7,  6,  5,  3,  1\}\{12,  10,  9,  7,  3,  1\}\{11,  10,  9,  6,  5,  1\}\\ 139& 10&
\{11,  6,  5,  4,  2\}\{10,  9,  8,  5,  4,  2\}\{13,  12,  9,  7,  4,  3\}\{13,  12,  11,  7,  5,  3\}\{8,  7,  6,  5,  4,  2\}\\ \hline
149	&1&	\{13, 10,  6,  5,  1\}\\ \hline
	&	&\{10, 9, 8, 7, 2, 1\}\{11, 9,  8, 7, 3, 1\}\{13, 12, 10, 7, 6, 1\}\{12, 11, 9, 8, 6, 1\}\{11,  8,  6,  4,  3,  2\}\\ 151& 13&\{11,  10,  7,  6,  4,  2\}
\{12,  9,  7,  6,  5,  2\}\{12,  11,  9,  8,  7,  2\}\{12,  10,  7,  6,  4,  3\}\{12,  11,  10,  9,  8,  3\}\\ & &\{13,  12,  8,  7,  6,  4\}\{13,  10,  9,  8,  7,  5\}\{10,  9,  8,  4,  3,  2\}\\ \hline
&	&\{12, 9, 6, 3, 2, 1\}\{11, 7, 5, 4, 2, 1\}\{13, 8, 5, 4, 2, 1\}\{10,  7,  6,  5,  3,  1\}\{13,  11,  9,  7,  3,  1\}\\ 157& 13&
\{13,  10,  7,  5,  4,  1\}\{12,  9,  8,  5,  3,  2\}\{12,  11,  8,  7,  5,  2\}\{13,  9,  8,  6,  4,  3\}\{12,  10,  9,  8,  6,  3\}\\ &&\{13, 10,  8,  6,  5,  4\}\{13,  12,  11,  10,  6,  3\}\{11,  10,  9,  8,  3,  1\}\\ \hline¡¡
	&&	\{11,  10,  7,  5,  3,  2\}\{13,  12,  11,  5,  3,  2\}\{9,  8,  6,  5,  4,  2\}\{11,  8,  7,  5,  4,  2\}\{13,  12,  7,  6,  4,  2\}\\ 163& 9&
\{11,  10,  9,  8,  6,  2\}\{11,  9,  7,  5,  4,  3\}\{13,  12,  11,  10,  9,  6\}¡¡\{13,  10,  8,  7,  5,  2\}\\ \hline¡¡
	&&\{12,  9,  6,  4,  2,  1\}\{13,  11,  6,  5,  2,  1\}\{12,  9,  7,  4,  3,  1\}\{12,  11,  10,  7,  3,  1\}\{13,  12,  7,  6,  4,  1\}\\ 167& 9	&
\{12,  11,  10,  9,  6,  2\}\{11,  9,  7,  5,  4,  3\}\{13,  12,  11,  10,  9,  6\}	\{12,  10,  7,  5,  3,  2\} \\ \hline
173&	1&	\{13, 12, 9, 7, 6\}\\ \hline
	&	&\{12,  11,  9,  3,  2,  1\}\{12,  10,  6,  5,  2,  1\}\{13,  8,  7,  5,  2,  1\}\{8,  7,  6,  5,  3,  1\}\{10,  9,  7,  6,  3,  1\}\\ 179& 9&
\{12,  8,  7,  5,  4,  3\}\{13,  11,  10,  7,  4,  3\}\{13,  11,  10,  7,  6,  5\}\{12,  11,  7,  6,  4,  1\}\\ \hline
	&	&\{8,  5,  4,  3,  2,  1\}\{13,  11,  7,  4,  2,  1\}\{9,  8,  7,  6,  2,  1\}\{13,  11,  10,  8,  3,  1\}\{13,  12,  10,  6,  4,  1\}\\ 181& 12&
\{11,  9,  6,  5,  3,  2\}\{12,  11,  10,  7,  3,  2\}\{13,  9,  8,  7,  5,  2\}\{13,  12,  10,  8,  7,  2\}\{13,  11,  10,  9,  7,  3\}\\&&\{12,  10,  8,  7,  5,  4\}\{9,  7,  5,  4,  3,  2\}\\ \hline
	&	&\{10,  7,  5,  3,  2,  1\}\{11,  7,  6,  4,  2,  1\}\{13,  10,  7,  4,  3,  1\}\{13,  11,  7,  6,  3,  1\}\{13,  12,  11,  9,  4,  1\}
\\ 191& 11&\{12,  10,  9,  6,  3,  2\}\{13,  8,  7,  6,  5,  4\}\{11,  10,  9,  7,  5,  4\}\{13,  12,  10,  8,  5,  4\}\{12,  11,  10,  8,  6,  4\}\\&& \{13,  12,  7,  6,  3,  2\}\\ \hline
	& &	\{12,  11,  10,  8,  2,  1\}\{11,  10,  8,  4,  3,  1\}\{13,  11,  5,  4,  3,  2\}\{10,  8,  7,  4,  3,  2\}\{11,  10,  8,  6,  3,  2\}\\ 193& 7&\{12,  10,  9,  8,  7,  2\}\{13,  11,  8,  7,  6,  5\}\\ \hline
	&	&\{11,  9,  5,  3,  2,  1\}\{12,  9,  8,  6,  2,  1\}\{12,  11,  10,  7,  2,  1\}\{13,  9,  7,  5,  3,  1\}\{11,  8,  7,  6,  3,  1\}\\ 197& 11&
\{12,  8,  7,  5,  4,  2\}\{13,  9,  7,  6,  5,  2\}\{12,  11,  7,  5,  4,  3\}\{12,  11,  10,  6,  4,  3\}\{13,  11,  10,  8,  5,  4\}\\&&\{13,  11,  8,  5,  3,  2\}\\ \hline
¡¡&		&\{13,  9,  4,  3,  2,  1\}\{13,  11,  8,  5,  2,  1\}\{13,  10,  8,  6,  4,  1\}\{11,  10,  9,  8,  7,  1\}\{12,  8,  6,  5,  4,  2\}\\ 199& 9&\{12,  10,  7,  6,  5,  2\}
\{10,  9,  8,  7,  5,  2\}\{12,  11,  10,  6,  5,  3\}\\&&\{12,  11,  10,  6,  4,  2\}\\ \hline
\vdots & & \ \ \ \ \ \ \ \ \ \ \ \ \ \ \ \ \ \ \ \ \ \ \ \ \ \ \ \ \ \ \ \ \ \vdots \\ \hline
$\geq22787$ && no privileged coalitions \ \ \ \ \ (proven) \\ \hline
\end{tabular*}
$$}
¡¡
¡¡
¡¡
¡¡
¡¡
¡¡

¡¡
¡¡
¡¡
¡¡
¡¡
¡¡
¡¡
¡¡
¡¡
¡¡
¡¡
¡¡
¡¡
¡¡
¡¡
¡¡
¡¡
¡¡
¡¡
¡¡	
¡¡
¡¡
¡¡
¡¡
¡¡
¡¡
¡¡
¡¡
¡¡

\end{document}